\documentclass[times]{rncauth}

\usepackage{moreverb}

\usepackage[dvips,colorlinks,bookmarksopen,bookmarksnumbered,citecolor=red,urlcolor=red]{hyperref}

\newcommand\BibTeX{{\rmfamily B\kern-.05em \textsc{i\kern-.025em b}\kern-.08em
T\kern-.1667em\lower.7ex\hbox{E}\kern-.125emX}}

\def\volumeyear{2010}

\newtheorem{definition}{Definition}
\newtheorem{assumption}{Assumption}
\newtheorem{lemma}{Lemma}
\newtheorem{corollary}{Corollary}
\newtheorem{example}{Example}
\newtheorem{proposition}{Proposition}

\newcommand{\be}{\begin{equation}}
\newcommand{\ee}{\end{equation}}
\newcommand{\bd}{\begin{displaymath}}
\newcommand{\ed}{\end{displaymath}}
\newcommand{\ba}{\begin{array}}
\newcommand{\ea}{\end{array}}
\newcommand{\bea}{\begin{eqnarray}}
\newcommand{\eea}{\end{eqnarray}}
\newcommand{\R}{\mathbb{R}}
\newcommand{\Q}{\mathbb{Q}}

\newcommand{\Hi} {{\cal H}^\infty}

\newcommand{\C}{\mathbb{C}}
\newcommand{\Cp}{\mathbb{C}_{+}}
\newcommand{\Cn}{\mathbb{C}_{-}}
\newcommand{\Zp}{\mathbb{Z}_{+}}

\newcommand{\N}[2]{\mathbb{N}_{#1}^{#2}}
\newcommand{\w}{\omega}

\begin{document}

\runningheads{S.~Gumussoy}{Factorization of SISO TDS and FIR Structure of Their $\Hi$ Controllers}

\title{Coprime Inner/Outer Factorization of SISO Time-Delay Systems and FIR Structure of Their Optimal H-Infinity Controllers}

\author{S.~Gumussoy\corrauth}

\address{Dept. of Computer Science, K.U. Leuven, Celestijnenlaan 200A, 3001 Heverlee, Belgium.}

\corraddr{Dept. of Computer Science, K.U. Leuven, Celestijnenlaan 200A, 3001 Heverlee, Belgium. Email: suat.gumussoy@cs.kuleuven.be}

\cgs{This work is supported by the Belgian Programme on Interuniversity Poles of Attraction, initiated by the Belgian State, Prime Ministeries Office for Science, Technology and Culture, the Optimization in Engineering Centre OPTEC of the K.U.Leuven, and the project STRT1-09/33 of the K.U.Leuven Research Foundation.}

\begin{abstract}
The approach in Foias \emph{et al.} (1996) is one of the well-developed methods to design \mbox{H-infinity} controllers for general infinite dimensional systems. This approach is applicable if the plant admits a special coprime inner/outer factorization. We give the largest class of single-input-single-output (SISO) time delay systems for which this factorization is possible and factorize the admissible plants. Based on this factorization, we compute the optimal \mbox{H-infinity} performance and eliminate unstable pole-zero cancellations in the optimal \mbox{H-infinity} controller. We extend the results on the finite impulse response (FIR) structure of optimal \mbox{H-infinity} controllers by showing that this structure appears not only for plants with input/output (I/O) delays, but also for general SISO time-delay plants.
\end{abstract}


\maketitle

\vspace{-6pt}

\section{Introduction}
\vspace{-2pt}
In control applications, robust controllers are desired to achieve stability and performance requirements under model uncertainties and exogenous disturbances \cite{ZhouBook}. The design requirements are usually defined in terms of $\Hi$ norms of the closed-loop functions including the plant, the controller and weights for uncertainties and disturbances. The \emph{optimal $\Hi$ controller design} finds the \emph{optimal $\Hi$ controller} stabilizing the plant and achieving the minimum $\Hi$ norm for the closed-loop system.
\medskip

The optimal $\Hi$ controller is designed for linear rational multi-input-multi-output (MIMO) systems based on Riccati and linear matrix inequality (LMI) methods by \cite{DGKF} and \cite{GAIJRNC94}. The optimal $\Hi$ controller is designed for different types of linear time-invariant time-delay systems.
\begin{enumerate}
	\item \emph{The rational SISO system with an I/O time-delay}: The optimal $\Hi$ controller is designed in \cite{ZKSCL87} \cite{FTZTAC86} using operator theoretic methods; see also \cite{SSCL89} \cite{OSCL90} and their references. State-space solution to the same problem is given in \cite{TSICON97} and \cite{MZTAC00}. Notably \cite{MZTAC00} showed the FIR structure of the optimal $\Hi$ controller.
	 \item  \emph{The rational MIMO system with I/O time-delays}: This plant is the generalization of the previous plant to the MIMO case. The system with the single delay in the feedback loop is considered in \cite{MeinsmaTAC2002} based on the $J$-spectral factorization approach and the FIR structure of dead-time compensators is shown. \cite{MMTAC04} extended this approach and illustrated the FIR structure for multiple I/O time-delays.
	 \item \emph{The infinite dimensional SISO system admitting a coprime-inner/outer factorization}:
The optimal $\Hi$ control problem is solved in \cite{FOT,TOTAC95} based on the parameterization of the solution of an interpolation problem due to Adamjan, Arov and Krein (AAK) \cite{AAK}. This technique is applied to \emph{stable} pseudorational plants in \cite{KYSCL05} and to the plants in a cascade connection of a rational MIMO plant and a scalar infinite dimensional inner function \cite{KPhD05, KashimaCDC2006, KashimaTAC2007, KashimaTAC2008}. Note that these techniques assume that the plant has a special factorization such as a coprime-inner/outer factorization or a series connection of a rational plant and an inner function.
\end{enumerate}

 Using the method in \cite{FOT, TOTAC95} based on AAK theory, we can design optimal $\Hi$ controllers for a large class of SISO infinite dimensional systems including SISO time-delay systems. The only requirement is that the plant should have a \emph{coprime-inner/outer factorization} and the factorization terms should be obtained. In this paper, we give conditions for \emph{general} SISO time-delay systems to check whether the coprime-inner/outer factorization in \cite{FOT,TOTAC95} is feasible and we factorize the admissible plants.

It is well known that optimal $\Hi$ controllers have an feedback connection of a rational transfer function and a FIR block for rational MIMO systems with I/O time-delays \cite{MZTAC00, MeinsmaTAC2002, MirkinTAC2003, MMTAC04}. The \emph{h-truncation} operator is introduced which truncates the impulse response to its restriction on a finite support and using this operator, the FIR structure of optimal $\Hi$ controllers is shown. In AAK theory based methods \cite{KashimaTAC2007, KashimaTAC2008}, the \emph{h-truncation} operator is extended to the \emph{generalized m-truncation} operator which generates an FIR system when the inner function is a \emph{single delay} and the outcome is not an FIR system for general inner functions. In this paper, we show that the FIR structure of optimal $\Hi$ controllers appears for not only systems with I/O time-delays, but also for more general SISO time-delay systems. We also prove that \emph{the FIR structure is due to the unstable pole-zero cancellations in the controller} and any controller with time-delay operators and polynomials having such cancellations has an FIR structure.

The structure of the article is as follows. In Section~\ref{section:prework} we introduce the necessary definitions and the assumptions. Main results are given in the next two sections. Section~\ref{sec:class} describes the coprime-inner/outer factorization of SISO time-delay systems. Section~\ref{sec:implement} computes the optimal $\Hi$ performance and gives the FIR structure of optimal $\Hi$ controllers. Section~\ref{sec:examples} is devoted to the numerical examples. In Section~\ref{sec:conc} some concluding remarks are presented.
\subsection*{Notations} The notations are
as follows:
\begin{tabbing}
  \= $j$\hspace{2.8cm} \=:  the imaginary identity, \\
  \> $\C,\ \R,\ \Q,\ \Zp$ \>: sets of complex, real, rational and positive integer numbers, \\
  \> $\Cp,\ \Cn$ \>: closed right and open left half complex planes, \\
  \> $\mathcal{H}^\infty(\Cp^o),\ \mathcal{H}^2(\Cp^o)$ \>: Hardy spaces of essentially bounded functions and square integrable functions \\
  \>  \> \ \ on the imaginary axis with bounded analytical extensions on open right half \\
  \>  \> \ \  complex plane,\\
  \> $\mathbb{N}, \N{n_1}{n_2}$ \>: set of natural numbers and set of natural numbers from $n_1$ to $n_2$, \\
  \> $\textrm{deg}\ p$ \> : degree of the polynomial $p$, \\
  \> $\sigma_{\min}(A)$ \> : the minimum singular value of a matrix A,\\
  \> $\|F\|_\infty$ \> : the supremum of the maximum singular value of the stable transfer function  \\
   \>  \> \ \ $F(j\w)$ $\forall\w\in\R$.
\end{tabbing}

\vspace{-6pt}

\section{Definitions and assumptions} \label{section:prework}
\vspace{-2pt}

We define the inner and outer functions which are frequently used in the paper.
\begin{definition} \label{def:inner_outer}
A function $m(s)$ is called \emph{inner} if $|m(s)|\leq 1$ for all $s\in\Cp$ and $|m(j\w)|=1$ almost everywhere $\w\in\R$. A function $N(s)\in\mathcal{H}^\infty(\Cp^o)$ is \emph{outer} if and only if the subspace $N \mathcal{H}^2(\Cp^o)$ is dense in $\mathcal{H}^2(\Cp^o)$.
 \end{definition}

Outer functions are generalization of minimum phase functions and they have no poles and zeros inside the open right half complex plane. For further information, see \cite{FOT}.

Given a SISO plant $P(s)$, the optimal $\Hi$ control problem is defined as
\be \label{eq:wsm1}
\gamma_{opt}=\inf_{C\ stab.\ P}\left\| \left[\begin{array}{c}
  W_1(s)(1+P(s)C(s))^{-1} \\
  W_2(s)P(s)C(s)(1+P(s)C(s))^{-1}
\end{array} \right]\right\|_\infty
\ee
where $W_1(s)$ and $W_2(s)$ are weights and rational transfer functions. The optimal $\Hi$ controller \emph{robustly stabilizes} the closed-loop system under model uncertainties represented by a frequency dependent bound $W_2(s)$ and achieves \emph{robust performance} such as good tracking over a low frequency range defined by the weight function $W_1(s)$. \cite{FOT,TOTAC95} computed the optimal $\Hi$ performance $\gamma_{opt}$ and the optimal $\Hi$ controller $C_{opt}$ of this problem for SISO plants admitting a coprime-inner/outer factorization
\be \label{Pfact}
P(s)=\frac{m_n(s)N_o(s)}{m_d(s)}
\ee
where $m_n(s)$ is inner, infinite dimensional and $m_d(s)$ is a rational inner transfer function and $N_o(s)$ is outer, possibly infinite dimensional. Note that the plant $P(s)$ is a general SISO plant. We will give conditions for SISO time-delay systems for which this factorization (\ref{Pfact}) is possible and obtain factorization terms in the next section.

\medskip

The transfer function of SISO time-delay systems is defined as a ratio of \emph{quasi-polynomials}. A \emph{quasi-polynomial} is a generalization of a polynomial and frequently used in transfer function representation of time-delay systems \cite{WimBook}.
\begin{definition} \label{def:qp} Let $q_i(s)$ for $i\in\N{1}{v}$ be polynomials with real coefficients and $h_i$ are non-negative real numbers in an ascending order. A function of the form
\be \label{qp}
q(s)=\sum_{i=1}^v q_i(s)e^{-h_i s}
\ee
is a \emph{quasi-polynomial} where $\textrm{deg}\ q_1 \geq \textrm{deg}\ q_i$ for $i\in\N{2}{v}$. Moreover, a quasi-polynomial is a \emph{neutral} type if there exists at least one $i\in\N{2}{v}$ such that $\textrm{deg}\ q_1 = \textrm{deg}\ q_i$, otherwise, it is a \emph{retarded} quasi-polynomial.
 \end{definition}
 The retarded and neutral quasi-polynomials have infinitely many roots in $\C$ \cite{HV93}. The asymptotic root chains of retarded quasi-polynomials extend to the infinity in real and imaginary parts only inside $\Cn$ whereas the asymptotic root chains of neutral quasi-polynomials contain at least one asymptotic chain of roots extending to the infinity parallel to the imaginary axis. We define the roots of a quasi-polynomial $q(s)$ inside $\Cp$ and $\Cn$ as \emph{unstable} and \emph{stable} roots of $q(s)$.
\medskip

We represent general SISO time-delay systems by the transfer function,
\be \label{P}
P(s)=\frac{q_n(s)}{q_d(s)}=\frac{\sum_{i=1}^{v_n}q_{n,i}(s)e^{-h_{n,i}
s}}{\sum_{k=1}^{v_d}q_{d,k}(s)e^{-h_{d,k} s}}
\ee
where $q_n(s)$ and $q_d(s)$ are quasi-polynomials given in Definition~\ref{def:qp}. Any realizable proper and causal transfer function $P(s)$ also satisfies the necessary conditions, $\textrm{deg}\ q_{n,1} \leq \textrm{deg}\ q_{d,1}$ and $h_{n,1} \geq h_{d,1}$.
\medskip

We assume that quasi-polynomials of the plant $P(s)$ (\ref{P}), $q_n(s)$ and $q_d(s)$, satisfy the following assumptions.
\begin{assumption} \label{assumption_rationaldelay}
Time-delays of $q_n(s)$ and $q_d(s)$ are rational numbers, i.e., , $h_{n,i}, h_{d,k}\in\Q$ for $i\in\N{1}{v_n}$ and $k\in\N{1}{v_d}$.
\end{assumption}
\begin{assumption} \label{assumption_asymptoticchain}
The asymptotic root chains of $q_n(s)$ and $q_d(s)$ do not approach the imaginary axis.
\end{assumption}

By the first assumption, we consider SISO time-delay systems with commensurate delays. This assumption is necessary to characterize the behavior of asymptotic root chains of quasi-polynomials using numerical polynomial root finding algorithms and it is not restrictive in practical control applications. The second assumption excludes a particular case, time-delay systems with poles or zeros asymptotically approaching the imaginary axis. The stabilization of these systems is an on-going research area and it is properties are not fully understood \cite{RSRJDE08,BFPCDC09}. As we shall see in the factorization of time-delay systems, this case also requires the computation of infinitely many imaginary axis roots which is numerically not possible.
\medskip

We give definitions of an \emph{asymptotic polynomial} and a \emph{conjugate quasi-polynomial} needed in the next section.
\begin{definition} \label{def:p}
Let $q(s)$ be a quasi-polynomial in Definition \ref{def:qp}. A function of the form \mbox{$p(s)=\sum_{i=1}^v p_i s^{n_i-n_1}$} is an \emph{asymptotic polynomial} of $q(s)$ where $ p_i=\lim_{s\rightarrow\infty}\frac{q_i(s)}{q_1(s)}$, $h_i=\frac{n_i}{N}$ and $N, n_i\in\Zp$, $p_i\in\R$ for~$i\in\N{1}{v}$.
\end{definition}

\begin{definition} \label{def:qp_conj}
Let $q(s)$ be a quasi-polynomial in Definition \ref{def:qp}. A function of the form
\mbox{$\bar{q}(s)=-q(-s)e^{-h_v s}$} is \emph{conjugate quasi-polynomial} of $q(s)$.
\end{definition}

In order to compute the optimal $\Hi$ performance and obtain the optimal $\Hi$ controller using the method in \cite{FOT,TOTAC95} we need to factorize SISO time-delay systems (\ref{P}) as in (\ref{Pfact}). As we shall see, this is not possible for all SISO time-delay systems. We classify SISO time-delay systems (\ref{P}) admitting the factorization (\ref{Pfact}) and factorize the admissible plants as in (\ref{Pfact}) in the next section.

\section{Coprime Inner/Outer Factorization of SISO Time-Delay Systems} \label{sec:class}

The main difficulty in the approach of \cite{FOT,TOTAC95} is to factorize a SISO time-delay system (\ref{P}) in the coprime-inner/outer form (\ref{Pfact}). For simple cases, this factorization is easy to do, i.e.,

\bd
P(s)=\left(\frac{s-1}{s-2}\right)e^{-s} \quad \Rightarrow \quad m_n(s)=\left(\frac{s-1}{s+1}\right)e^{-s},\;\; m_d(s)=\frac{s-2}{s+2},\;\; N_o(s)=\frac{s+1}{s+2}.
\ed However, \emph{in the general case, it might be difficult to find the inner-outer factorization} (p.23 \cite{FOT}). Since the inner function $m_d(s)$ in (\ref{Pfact}) is rational, the plant $P(s)$ (\ref{P}) (equivalently, the quasi-polynomial $q_d(s)$) has to have finitely many poles (roots) inside $\Cp$. By the following Lemma, We determine whether a given quasi-polynomial has finitely many roots inside $\Cp$.
\begin{lemma} \label{lemma:finitezeros}		
Let $q(s)$ be a quasi-polynomial given in Definition \ref{def:qp} and satisfying Assumption \ref{assumption_rationaldelay}, \ref{assumption_asymptoticchain}. Then the quasi-polynomial $q(s)$ has finitely many roots inside $\Cp$ if and only if $q(s)$ is a retarded quasi-polynomial or $q(s)$ is a neutral quasi-polynomial and the magnitude of roots of its asymptotic polynomial is greater than~$1$.
\end{lemma}
\begin{proof}
By Definition~\ref{def:qp}, the quasi-polynomial $q(s)$ is either retarded or neutral type. A retarded quasi-polynomial has finitely many roots inside $\Cp$ \cite{HV93}. The neutral quasi-polynomial has finitely many roots inside $\Cp$ if and only if all its asymptotic root chains extending to the infinity lie inside $\Cn$. The asymptotic root chains of a neutral quasi-polynomial either lie inside $\Cn$ or extend to the infinity in the imaginary part and to a constant value in the real part $\sigma_o$ \cite{WimBook}. Therefore, it is enough to show that the real part of any asymptotic root chains extending to the infinity parallel to the imaginary axis is negative, i.e., $\sigma_o<0$. For large values of the imaginary part $\w$, the behavior of these root chains are characterized by the asymptotic polynomial $p(s)$, i.e., for any fixed $\sigma_o\in\R$ and $\epsilon>0$, there exists $\w_*>0$ such that
\be \label{lemasympol}
\left|\frac{q(s)}{q_1(s)e^{-h_1 s}}
-p\left(e^{-\frac{s}{N}}\right)\right|_{s=\sigma_0+j\w}=
\left|\sum_{i=1}^v\left(\frac{q_i(s)}{q_1(s)}-p_i\right)e^{-(h_i-h_1)s}\right|_{s=\sigma_0+j\w}
<\epsilon\quad \textrm{for}\quad \w>w_*.
\ee This establishes the connection between an asymptotic chain root of $q(s)$, $r_{q,i}$ and the corresponding root of its asymptotic polynomial $p(s)$, $r_{p,i}$ as
\bd
r_{q,i}=\sigma_0+j\w_i \Longleftrightarrow r_{p,i}=e^{-(\sigma_0+j\w_i)/N}.
\ed The real part of $r_{q,i}$ is negative, $\sigma_o<0$, if and only if the magnitude of the root of the asymptotic polynomial $p(s)$ is greater than $1$, $|r_{p,i}|>1$. The assertion follows.
\end{proof}

Given a neutral quasi-polynomial with infinitely many roots inside $\Cp$, its conjugate quasi-polynomial may have finitely many roots inside $\Cp$. This class of neutral quasi-polynomials plays an important role in the plant factorization (\ref{Pfact}). The following Corollary allows us to classify such neutral quasi-polynomials.
\begin{corollary} \label{cor:conjfinzeros}		
Let $q(s)$ be a quasi-polynomial given in Definition \ref{def:qp} and satisfying Assumption \ref{assumption_rationaldelay}, \ref{assumption_asymptoticchain}. The conjugate quasi-polynomial $\bar{q}(s)$ has finitely many roots inside $\Cp$ if and only if the magnitude of roots of the asymptotic polynomial $p(s)$ of $q(s)$ is smaller than~$1$.
\end{corollary}
\begin{proof} The roots of quasi-polynomials $\bar{q}(s)$ and $q(-s)$ in the complex plane are same. By following the steps of the proof in Lemma~\ref{lemma:finitezeros} for $q(-s)$, we obtain
\bd
r_{\bar{q},i}=\sigma_0+j\w_i \Longleftrightarrow r_{p,i}=e^{(\sigma_0+j\w_i)/N}
\ed which is the reciprocal of the case in Lemma~\ref{lemma:finitezeros}. The assertion follows.
\end{proof}

It is not easy to see that whether a quasi-polynomial or its conjugate has finitely many roots inside $\Cp$ without computing some of right-most roots of the quasi-polynomial. Lemma~\ref{lemma:finitezeros} and Corollary~\ref{cor:conjfinzeros} answer this question by only computing the magnitude of the roots of its asymptotic polynomial.
\begin{example} \label{ex_qroots}
Consider the following quasi-polynomials,
 \[
q_1(s)=3s+0.5+(2s+7)e^{-1.5s}+(s-1)e^{-2s}\quad\textrm{and}\quad q_2(s)=s+3+(2s-2)e^{-0.4s}.
\]
The quasi-polynomial $q_1(s)$ is neutral and its corresponding asymptotic polynomial is $p_1(s)=1/3s^4+2/3s^3+1$.  The magnitudes of their roots are $\{1.6796, 1.0312\}$ and greater than $1$. By Lemma~\ref{lemma:finitezeros}, $q_1(s)$ has finitely many roots inside $\Cp$.  Figure~\ref{fig:q1} shows that there are four roots of $q_1(s)$ inside $\Cp$.

The asymptotic polynomial of the neutral quasi-polynomial $q_2(s)$ is $p_2(s)=2s+1$. The magnitude of its root is smaller than $1$ and by Lemma~\ref{lemma:finitezeros}, $q_2(s)$ has infinitely many roots inside $\Cp$. This is illustrated in Figure~\ref{fig:q2} by computed roots of $q_2(s)$ shown with star markers. On the other hand, since the root is smaller than $1$, the conjugate quasi-polynomial $\bar{q}_2(s)$ has finitely many roots inside $\Cp$ by Corollary~\ref{cor:conjfinzeros}. Figure~\ref{fig:q2} shows the roots of $\bar{q}_2(s)$ with circle markers and
there is one root inside $\Cp$.

\begin{figure}[!h]
    \begin{minipage}[t]{0.48\textwidth}
        \vspace{0pt}
        \includegraphics[width=\linewidth]{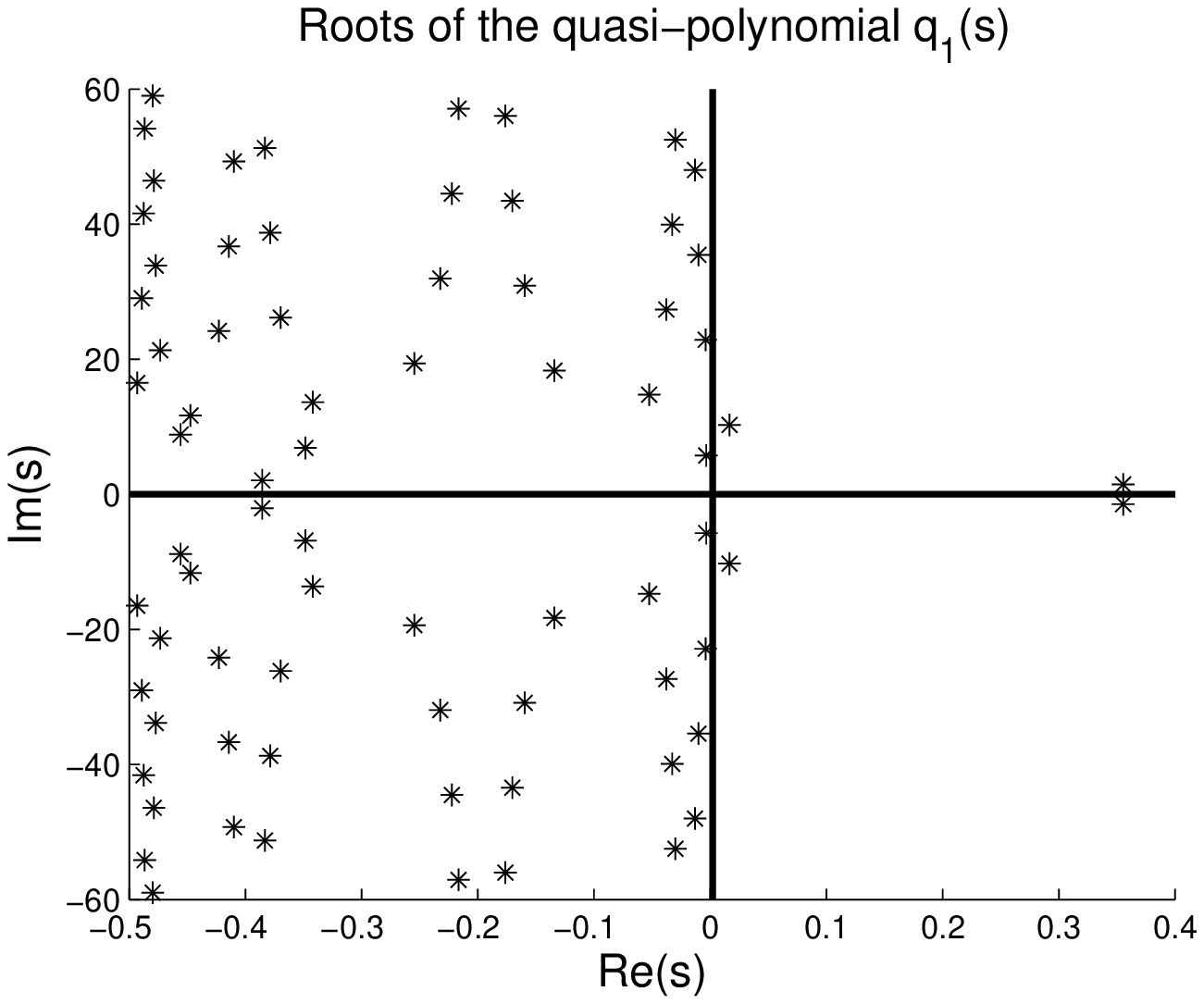}
        \caption{\label{fig:q1} The roots of the quasi-polynomial $q_1(s)$.}
    \end{minipage}
\hfill
    \begin{minipage}[t]{0.49\textwidth}
        \vspace{0pt}\raggedright
        \includegraphics[width=\linewidth]{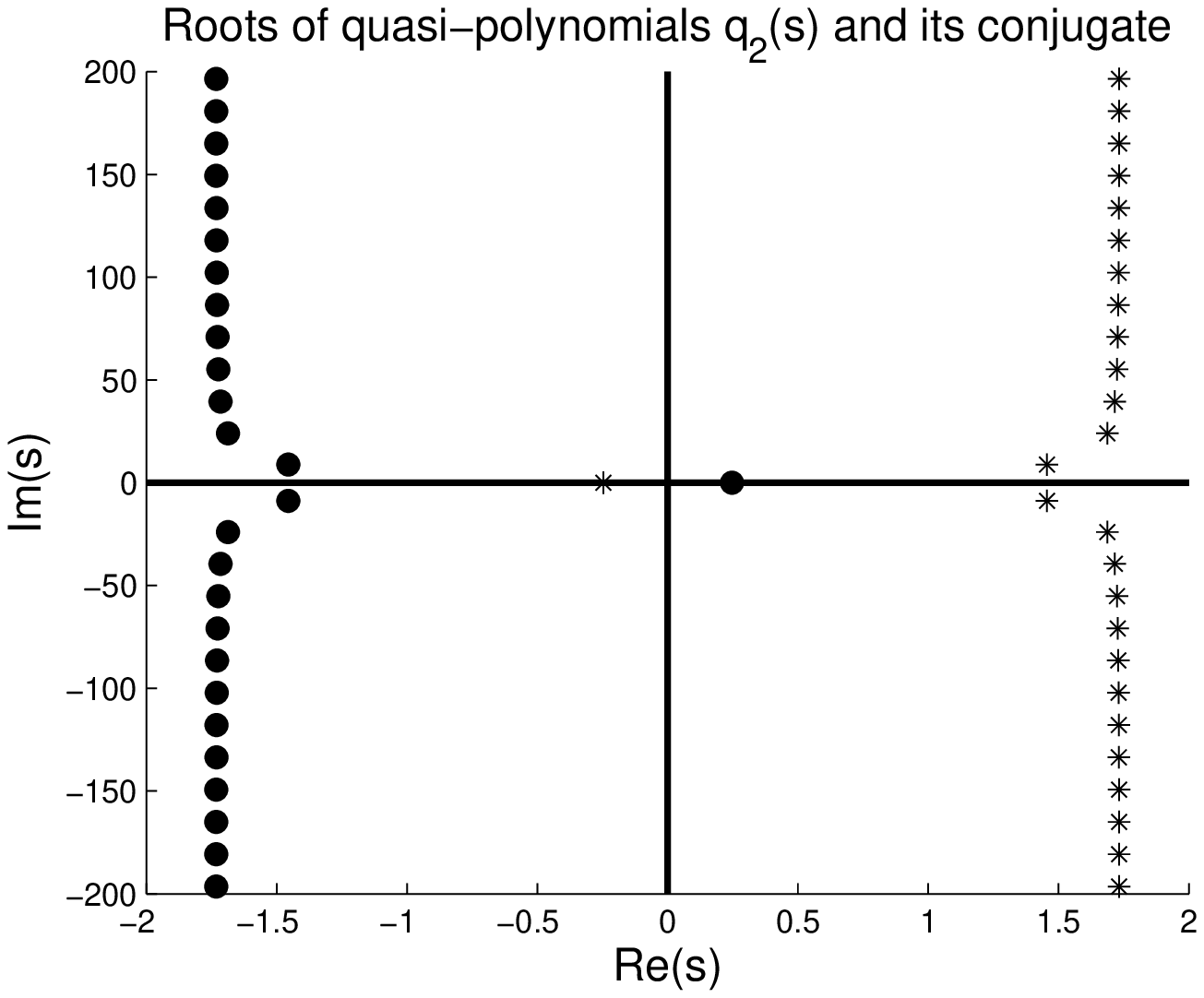}
        \caption{\label{fig:q2} The roots of quasi-polynomials $q_2(s)$ (with stars) and $\bar{q}_2(s)$ (with circles).}
   \end{minipage}
\end{figure}
\end{example}

We now give the key lemma for the inner-outer factorization of quasi-polynomials in a closed-form.
\begin{lemma} \label{lemma:qfact}			
Let $q(s)$ be a quasi-polynomial given in Definition \ref{def:qp} and satisfying Assumption \ref{assumption_asymptoticchain}. It admits an inner-outer factorization if $q(s)$ or $\bar{q}(s)$ has finitely many roots inside $\Cp$.
\end{lemma}
\begin{proof}
Let $q(s)$ be a quasi-polynomial with finitely many roots inside $\Cp$.  We construct a rational inner function $m_q(s)$  whose zeros are the roots of $q(s)$ inside $\Cp$. Then, we can factorize the quasi-polynomial $q(s)$ as
\be \label{qfactF}
q(s)=m(s)N(s)=m_q(s) \left(\frac{q(s)}{m_q(s)}\right).
\ee  Note that $m(s)$ is an inner function by construction of $m_q(s)$ and $N(s)$ is an outer function with no zeros or poles inside $\Cp$ since all zeros of $m_q(s)$ are cancelled by the roots of $q(s)$.

Similarly, let $\bar{q}(s)$ be a quasi-polynomial with finitely many roots inside $\Cp$. Having constructed a rational inner function $m_{\bar{q}}(s)$  whose zeros are the roots of $\bar{q}(s)$ inside $\Cp$, we factorize $q(s)$ as
\be \label{qfactI}
q(s)=m(s)N(s)=\left(\frac{q(s)}{\bar{q}(s)}m_{\bar{q}}(s)\right) \left(\frac{\bar{q}(s)}{m_{\bar{q}}(s)}\right)
\ee  The magnitude of $m(s)$ on the imaginary axis is $1$. All roots of $\bar{q}(s)$ inside $\Cp$ are cancelled by the zeros of $m_{\bar{q}(s)}$, hence $m(s)$ is stable. Then $m(s)$ and $N(s)$ are inner and outer functions respectively. The assertions follow.
\end{proof}

The main result of this section allows us to classify SISO time-delay systems (\ref{P}) admitting a coprime-inner/outer factorization (\ref{Pfact}) in a closed-form.
\begin{proposition} \label{prop:Pfact}		
Let $q_n(s)$ and $q_d(s)$ be quasi-polynomials given by Definition \ref{def:qp} and satisfying Assumption~\ref{assumption_asymptoticchain}. The plant $P(s)$ (\ref{P}) admits a coprime-inner/outer factorization (\ref{Pfact}) if either of the following holds:
\be \label{prop:cond}
\left.\begin{array}{c}
C_1:\quad q_n(s)\ \textrm{and}\ q_d(s)\quad \\
\textrm{or} \\
C_2:\quad \bar{q}_n(s)\ \textrm{and}\ q_d(s)\quad
\end{array}\right\}\ \textrm{have finitely many roots inside}\ \Cp.
\ee
Moreover, the factorization for each case is
\be
\begin{array}{llll} \label{Ccases}
C_1: &\quad m_n(s)=m_{q_n}(s), &\quad m_d(s)=m_{q_d}(s), &\quad N_o(s)=\frac{q_n(s)}{m_{q_n}(s)}\frac{m_{q_d}(s)}{q_d(s)},\\
C_2: &\quad m_n(s)=m_{\bar{q}_n}(s)\frac{q_n(s)}{\bar{q}_n(s)}, &\quad m_d(s)=m_{q_d}(s), &\quad N_o(s)=\frac{\bar{q}_n(s)}{m_{\bar{q}_n}(s)}\frac{m_{q_d}(s)}{q_d(s)}
\end{array}
\ee
where $m_{q_n}(s)$, $m_{\bar{q}_n}(s)$, and $m_{q_d}(s)$ are inner functions and their zeros are the roots of $q_n(s)$, $\bar{q}_n(s)$, and $q_d(s)$ inside $\Cp$ respectively.
\end{proposition}
\begin{proof}
Since the coprime-inner/outer factorization in (\ref{Pfact}) requires a rational inner function $m_d(s)$, $q_d(s)$ in (\ref{P}) has to have finitely many roots inside $\Cp$. If $q_n(s)$ or $\bar{q}_n(s)$ has finitely many roots inside $\Cp$, by Lemma~\ref{lemma:qfact}, we can write the plant $P$ (\ref{P}) as
\bd
P(s)=\frac{q_n(s)}{q_d(s)}=\frac{m_n(s)N_n(s)}{m_d(s)N_d(s)}=\frac{m_n(s)N_o(s)}{m_d(s)}
\ed using the inner-outer factorizations of $q_n(s)$ and $q_d(s)$. The functions $m_n(s)$ and $m_d(s)$ are inner functions and $N_o(s)$ is an outer function as in (\ref{Pfact}). Since $q_d(s)$ is rational, it is factorized using (\ref{qfactF}). We get the factorizations in $C_1$ and $C_2$ at (\ref{Ccases}), when $q_n(s)$ is factorized using (\ref{qfactF}) and (\ref{qfactI}) respectively. The assertions follow.
\end{proof}

Proposition~\ref{prop:Pfact} give the conditions and the factorizations for the admissible plants. These conditions can be easily checked by Lemma~\ref{lemma:finitezeros} and Corollary~\ref{cor:conjfinzeros}. Note that we need to compute the roots of quasi-polynomials $q_d(s)$ and  $q_n(s)$ or $\bar{q}_n(s)$. There are available numerical methods and tools for the computation of roots of quasi-polynomials inside $\Cp$, see \cite{BIFTOOLManual20, BredaTraceDDE09, VZTAC09}.

\begin{example} \label{ex:fact}
We determine the admissible plants of the following SISO time-delay systems (\ref{P}),
\bd
\begin{array}{lll}
P_1(s)=\frac{s^2-2s+3+0.2se^{-s}}{s^3+1+e^{-1.5s}},  & & \vspace{2mm} P_2(s)=\frac{(s-1)e^{-0.2s}+(0.1s+1)e^{-0.3s}+(0.2s-3)e^{-s}}{3s+0.5+(2s+7)e^{-1.5s}+(s-1)e^{-2s}}, \\
P_3(s)=\frac{s+3+(2s-2)e^{-0.4s}}{s^2+se^{-0.2s}+5e^{-0.5s}}, & &
\end{array}
\ed
Both $q_n(s)$ and $q_d(s)$ of $P_1(s)$ are retarded quasi-polynomials and they have finitely many roots inside $\Cp$ by Lemma~\ref{lemma:finitezeros}. The plant $P_1(s)$ is an admissible plant by Proposition~\ref{prop:Pfact}. The quasi-polynomial at the denominator of $P_2(s)$ is considered in Example~\ref{ex_qroots} as $q_1(s)$ and it has finitely many roots inside $\Cp$. The numerator $q_n(s)$ of $P_2(s)$ is a neutral quasi-polynomial. Its asymptotic polynomial is $p_n(s)=0.2s^8+0.1s+1$ and the magnitudes of its roots are $\{1.2396, 1.2306, 1.2164, 1.2051\}$ which are greater than $1$. By Lemma~\ref{lemma:finitezeros}, $q_n(s)$ and  $q_d(s)$ have finitely many roots inside $\Cp$ and the plant $P_2(s)$ is admissible. The numerator of $P_3(s)$ is considered in Example~\ref{ex_qroots} as $q_2(s)$ and its conjugate $\bar{q}_n(s)$ has finitely many roots inside $\Cp$. The denominator of $P_3(s)$ is a retarded quasi-polynomial with  finitely many roots inside $\Cp$ by Lemma~\ref{lemma:finitezeros}. Since $\bar{q}_n$ and $q_d(s)$ have finitely many roots inside $\Cp$, we conclude that the plant $P_3(s)$ is admissible by Proposition~\ref{prop:Pfact}. The coprime-inner/outer factorizations of the admissible plants $P_1(s)$, $P_2(s)$ using $C_1$ in (\ref{Ccases}) and $P_3(s)$ using $C_2$ in (\ref{Ccases}) are given in Appendix~\ref{sec:app_fact}.
\end{example}

Our main result in this section, Proposition~\ref{prop:Pfact}, allows us to obtain the coprime-inner/outer factorization (\ref{Pfact}) of the admissible plants. Using factorization terms, $m_n(s)$, $m_d(s)$ and $N_o(s)$ (\ref{Ccases}), we compute the optimal $\Hi$ performance and give the structure of the optimal $\Hi$ controller in the next section.

\section{FIR Structure of Optimal H-Infinity Controllers} \label{sec:implement}

We can compute the optimal $\Hi$ performance $\gamma_{opt}$ for the optimal $\Hi$ controller design problem~(\ref{eq:wsm1}) using the coprime-inner/outer factorization terms of (\ref{Pfact}) of the SISO time-delay system $P(s)$~(\ref{P}) and the given weight functions $W_1(s)$,$W_2(s)$ \cite{TOTAC95}. As described in the paper, we sweep a scalar parameter $\gamma$ over an interval and compute the minimum singular value of a matrix, $\sigma_{\min}(M(\gamma))$, depending on the factorization terms $m_n(s)$, $m_d(s)$ and weight functions. The largest value of $\gamma$ where $\sigma_{\min}(M(\gamma))=0$ is the optimal $\Hi$ performance $\gamma_{opt}$.

Having computed $\gamma_{opt}$, we obtain the optimal $\Hi$ controller for the problem (\ref{eq:wsm1}) as
 \be \label{Copt}
C_{opt}(s)=\frac{m_d(s) E_{\gamma_{opt}}(s) F_{\gamma_{opt}}(s)  L_{\gamma_{opt}}(s) N_o^{-1}(s)}{1+m_n(s) F_{\gamma_{opt}}(s)L_{\gamma_{opt}}(s)}
\ee where $E_{\gamma_{opt}}(s)$, $F_{\gamma_{opt}}(s)$ and $L_{\gamma_{opt}}(s)$ are rational transfer functions depending on $\gamma_{opt}$ and see \cite{TOTAC95} for further details.

The optimal $\Hi$ controller (\ref{Copt}) contains the factorization terms  $m_n(s)$, $m_d(s)$ and $N_o(s)$. By construction in (\ref{Ccases}), these terms may have unstable pole-zero cancellations. Moreover, due to interpolation conditions, there are additional unstable pole-zero cancellations in optimal $\Hi$ controllers \cite{TOTAC95}. Unlike the finite dimensional case, the factorization terms are infinite dimensional because of time-delays and therefore exact cancellations are not possible. Our main result in Section~\ref{sec:fir} allows us to deal with unstable pole-zero cancellations in SISO time-delay systems.  In Section~\ref{sec:implement_copt}, by rearrangement of terms in the optimal $\Hi$ controller and using this main result, we eliminate unstable pole-zero cancellations in the optimal $\Hi$ controller (\ref{Copt}). This leads to optimal $\Hi$ controllers with FIR structures and extends the results in \cite{MZTAC00, MeinsmaTAC2002, MMTAC04} by showing that the FIR structure of $\Hi$ controllers appears not only for plants with I/O delays, but also for general SISO time-delay plants.

\subsection{FIR structure in unstable pole-zero cancellations of SISO time-delay systems} \label{sec:fir}

The unstable pole-zero cancellations in a transfer function can be eliminated for the finite dimensional case. This elimination is not possible for time-delay systems due to the irrational transfer functions. We consider SISO time-delay systems
\be \label{eq:G}
G(s)=\sum_{k=1}^v G_k(s)e^{-h_k s}
\ee where $G_k(s)$, $k\in\N{1}{v}$ are proper, rational transfer functions. We show that unstable pole-zero cancellations between a SISO time-delay system (\ref{eq:G}) and a rational SISO system can be transformed into a transfer function whose impulse response has a finite support, known as an FIR filter. The following proposition is a key result to eliminate unstable pole-zero cancellations in the optimal $\Hi$ controllers.
\begin{proposition} \label{prop:fir}
Let $G(s)$ be a SISO time-delay system (\ref{eq:G}), $G_0(s)$ be a bi-proper, rational system and $\mathcal{S}_z^+$ be the nonempty set of common $\mathbb{C}_+$ zeros of $G(s)$ and $G_0(s)$. If the transfer function $\frac{G(s)}{G_0(s)}$ is decomposed as $H(s)+F(s)$ such that $H(s)$ has no poles in $\mathcal{S}_z^+$ and the poles of $F(s)$ are the elements of $\mathcal{S}_z^+$, then $H(s)$ has no unstable pole-zero cancellations from $\mathcal{S}_z^+$ and $F(s)$ is an FIR filter with a support $t\in[0,h_v]$.
\end{proposition}
\begin{proof} For simplicity assume that $z_1,z_2,\dots,z_{n_z}\in\mathcal{S}_z^+$ are distinct. We decompose $\frac{G(s)}{G_0(s)}$ into two terms as
\be \label{eq:phi}
\frac{G(s)}{G_0(s)}=\sum_{i=1}^v \frac{G_i(s)}{G_0(s)}e^{-h_i s}=H(s)+F(s)=\left(\sum_{i=1}^v H_i(s)e^{-h_i s}\right)+\left(\sum_{i=1}^v F_i(s)e^{-h_i s}\right)
\ee where $F_i(s)$ is strictly proper and its poles are all elements of $\mathcal{S}_z^+$ obtained by the partial fraction, $\frac{G_i(s)}{G_0(s)}=H_i(s)+F_i(s)$ for $i\in\N{1}{n_z}$. By construction, $H(s)$ has no poles in $\mathcal{S}_z^+$ and $F(z_k)$ is finite for $k\in\N{1}{n_z}$. We prove the assertion if we can show that the inverse Laplace transform of $F(s)$ has a compact support and its support is equal to $[0,h_v]$. Since the function $F(s)$ is entire and it can be bounded by an exponential function as $|F(s)|\leq C e^{\delta|s|}$ for $\delta>0$, the support of its inverse Laplace transform is compact and inside $[-\delta,\delta]$ by the Paley-Wiener theorem \cite{Rudin1987}. In order to show that its impulse response $f(t)\equiv0$ for $t>h_v$, the inverse Laplace transform of $F(s)$ is written as
\bd
f(t)=\sum_{k=1}^{n_z}\left[\sum_{i=1}^v \textrm{Res}
(F_i(s))\big|_{s=z_k}\;e^{z_k(t-h_i)}u(t-h_i) \right]
\ed where $u(t)$ and $\textrm{Res(.)}$ are the unit step function and the residue of the function respectively. For $t>h_v$, we have
\bd
f(t)=\sum_{k=1}^{n_z}e^{z_k t}\left[\sum_{i=1}^v \textrm{Res}(F_i(s))\big|_{s=z_k}\;e^{-h_i z_k} \right].
\ed Using $\textrm{Res}(F_i(s))\big|_{s=z_k}=\frac{G_i(z_k)}{\textrm{Res}(G_0(s))\big|_{s=z_k}}$ for $i\in\N{1}{n_v}$ and $k\in\N{1}{n_z}$, the impulse response of $F(s)$ is equal to
\bd
f(t)=\sum_{k=1}^{n_z}\frac{e^{z_k t}}{\textrm{Res}(G_0(s))\big|_{s=z_k}}\left[\sum_{i=1}^v G_i(z_k)\;e^{-h_i z_k} \right]=\sum_{k=1}^{n_z}\frac{e^{z_k t}}{\textrm{Res}(G_0(s))\big|_{s=z_k}} G(z_k)\equiv0\ \ \textrm{for}\  t>h_v.
\ed The last equivalence follows from the fact that $z_k$, $k\in\N{1}{n_z}$ are the zeros of
$G(s)$, i.e., $G(z_k)=0$ and ${\textrm{Res}(G_0(s))\big|_{s=z_k}} \neq 0$.  By similar arguments, the results can be proven for common roots with multiplicities in $\mathcal{S}_z^+$. The assertion follows.
\end{proof}

We define the operator $\Phi$ which decomposes the transfer function $\frac{G(s)}{G_0(s)}$ as in Proposition \ref{prop:fir}, i.e.,
\bd
\Phi(G(s),G_0(s))=H(s)+F(s).
\ed The operator decomposes each rational SISO transfer functions $\frac{G_i(s)}{G_0(s)}$, $i\in\N{1}{v}$ in $\frac{G(s)}{G_0(s)}$ into two rational transfer functions $H_i(s)$ and $F_i(s)$ (\ref{eq:phi}) where the poles of $F_i(s)$ are all elements of $\mathcal{S}_z^+$. It combines $H_i(s)$ and $F_i(s)$, $i\in\N{1}{v}$ separately and obtains the transfer functions $H(s)$ and $F(s)$ respectively.

Note that the decomposition $\Phi(G(s),G_0(s))$ eliminates unstable pole-zero cancellations in $G(s)$ and $G_0(s)$ and by Proposition~\ref{prop:fir}, the cancellation terms are contained in an FIR term. This can be seen as the extension of exact pole-zero cancellations in a finite dimensional system to that in a time-delay system.

\begin{example} \label{ex:fir}
Consider the following SISO time-delay and the rational SISO systems,
\[
G(s)=\frac{(-18.5952s^2-27.8651s+18.5796)(s^3+1+e^{-1.5s})}
{s^5+15s^4+59s^3+97s^2+72s+20},
\ \  \ G_0(s)=\frac{s^2-1.2470s+1.1137}{s^2+1.2470s+1.1137}.
\] The common unstable zeros of $G(s)$ and $G_0(s)$ are $0.6235\pm0.8514j$. The unstable pole-zero cancelations can be eliminated by the decomposition operator $\Phi$ as
\begin{multline*}
\frac{G(s)}{G_0(s)}=\Phi(G(s),G_0(s))=F(s)+H(s) \\
\textrm{where}\quad  F(s)=\frac{(-0.1260s+0.3061)-(0.5588s+0.0810)e^{-1.5s}}{s^2-1.2470s+1.1137}.
\end{multline*}

The impulse response of $F(s)$ is given in Figure~\ref{fig:F}. As proven in Proposition~\ref{prop:fir}, the response has a finite support, i.e., $[0,1.5]$ and therefore $F(s)$ is an FIR.

\begin{figure}[h]
   \centering \includegraphics[width=8cm]{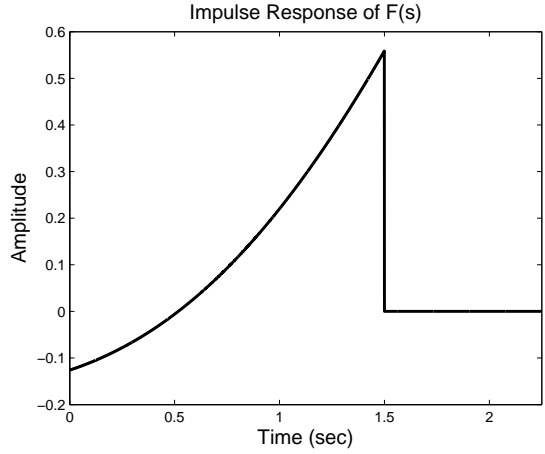}
   \caption{The impulse response of $F(s)$} \label{fig:F}
\end{figure}
\end{example}

By Proposition~\ref{prop:fir}, we see that any SISO time-delay system with unstable pole-zero cancellations has an FIR structure. Since $\Hi$ controllers have always have unstable pole-zero cancellations due to the interpolation conditions, it is easy to see why $\Hi$ controllers of time-delay systems have FIR structures. This structure is illustrated in the next section.

\subsection{FIR structure of optimal $\Hi$ controllers} \label{sec:implement_copt}

The coprime-inner/outer factorization of the admissible plants is given in (\ref{Ccases}) of Proposition~\ref{prop:Pfact}. Using the factorized terms $m_n(s)$, $m_d(s)$ and $N_o(s)$, we calculate the optimal $\Hi$ performance $\gamma_{opt}$ and the optimal $\Hi$ controller (\ref{Copt}) as in  \cite{FOT,TOTAC95}.  The admissible plants are composed of two groups: $q_n(s)$, $q_d(s)$ of $P(s)$ (\ref{P}) have finitely many $\Cp$ roots ($C_1$ case in (\ref{Ccases})) and  $\bar{q}_n(s)$, $q_d(s)$ of $P(s)$ (\ref{P}) have finitely many $\Cp$ roots ($C_2$ case in (\ref{Ccases})). We will eliminate unstable pole-zero cancellations for each case and give the FIR structure of corresponding optimal $\Hi$ controllers.

\subsubsection*{Optimal $\Hi$ controller structure when $q_n(s)$ and $q_d(s)$ has finitely many $\Cp$ roots}

By substituting the factorization terms for $C_1$ case of (\ref{Ccases}), the optimal $\Hi$ controller can be
written as,
\be \label{Copt1}
C_{opt}(s)=\frac{\left(E_{\gamma_{opt}}(s)F_{{\gamma}_{opt}}(s)m_{q_d}(s)L_{{\gamma}_{opt}}(s)\right)\left(\frac{q_d(s)}{m_{q_d}(s)}\right)\left(\frac{q_n(s)}{m_{q_n}(s)}\right)^{-1}}
{\left(1+m_{q_{n}}(s)F_{\gamma_{opt}}(s)L_{\gamma_{opt}}(s)\right)}.
\ee The optimal $\Hi$ controller has unstable pole-zero cancellations for the terms inside the second and third parenthesis in the numerator (\ref{Copt1}) because of the construction of the factorization terms. Due to the interpolation conditions, the zeros inside $\Cp$ of  $E_{\gamma_{opt}}(s)$ and $m_{q_d}(s)$ are cancelled by the term in the denominator of (\ref{Copt1}). These cancellations are eliminated and the structure of the $\Hi$ controller is obtained as the following.
\begin{enumerate}
    \item Factorize the rational transfer function  $E_{\gamma_{opt}}(s)F_{{\gamma}_{opt}}(s)m_{q_d}(s)L_{{\gamma}_{opt}}(s)$ as
    $\theta_n(s)\theta_d(s)$ where $\theta_n(s)$ is a bi-proper transfer function whose zeros are $\Cp$ zeros of $E_{\gamma_{opt}}(s)$ and $m_{q_d}(s)$,
    \item Rewrite the controller (\ref{Copt1}) as
	\be \label{CoptFF}
		 C_{opt}(s)=\left(\frac{\theta_d(s)q_d(s)}{m_{q_d}(s)}\right)\left(\frac{q_n(s)(1+m_{q_n}(s)F_{\gamma_{opt}}(s)L_{\gamma_{opt}}(s))}{\theta_n(s)m_{q_n}(s)}\right)^{-1}
	\ee where the zeros inside $\Cp$ of the denominator terms in each parenthesis are the zeros of the numerator terms.
    \item Eliminate unstable pole-zero cancellations by the decomposition operator $\Phi$,
 \bea
\nonumber H_n(s)+F_n(s)&=&\Phi(\theta_d(s)q_d(s),m_{q_d}(s)), \\
\nonumber H_d(s)+F_d(s)&=&\Phi(q_n(s)(1+m_{q_n}(s)F_{\gamma_{opt}}(s)L_{\gamma_{opt}}(s)),\theta_n(s)m_{q_n}(s)).
\eea
\end{enumerate}
 Then, the optimal $\Hi$ controller has no unstable pole-zero cancellations as,
\be \label{Coptform}
C_{opt}(s)=\frac{H_n(s)+F_n(s)}{H_d(s)+F_d(s)}
\ee where $F_n(s)$ and $F_d(s)$ are FIR filters by Proposition~\ref{prop:fir}.

\subsubsection*{Optimal $\Hi$ controller structure when $\bar{q}_n(s)$ and $q_d(s)$ has finitely many $\Cp$ roots}

Similarly, we substitute the factorization terms for $C_2$ case of (\ref{Ccases}) and rearrange the optimal $\Hi$ controller as
	\be \label{CoptIF}
		 C_{opt}(s)=\left(\frac{\theta_d(s)q_d(s)}{m_{q_d}(s)}\right)\left(\frac{(\bar{q}_n(s)+m_{\bar{q}_n}(s)q_n(s)F_{\gamma_{opt}}(s)L_{\gamma_{opt}}(s))}{\theta_n(s)m_{\bar{q}_n}(s)}\right)^{-1}
	\ee where  $E_{\gamma_{opt}}(s)F_{{\gamma}_{opt}}(s)m_{q_d}(s)L_{{\gamma}_{opt}}(s)$ is factorized  $\theta_n(s)\theta_d(s)$ as in the previous section. The unstable pole-zero cancellations are eliminated by the decomposition operator $\Phi$,
 \bea
\nonumber H_n(s)+F_n(s)&=&\Phi(\theta_d(s)q_d(s),m_{q_d}(s)), \\
\nonumber H_d(s)+F_d(s)&=&\Phi((\bar{q}_n(s)+m_{\bar{q}_n}(s)q_n(s)F_{\gamma_{opt}}(s)L_{\gamma_{opt}}(s)),\theta_n(s)m_{\bar{q}_n}(s)).
\eea

The optimal $\Hi$ controller for this case has the same form as in (\ref{Coptform}).

We conclude this section with few remarks. The FIR structure of optimal $\Hi$ controller (\ref{Coptform}) is the extension of well-known results for I/O time-delays \cite{MZTAC00} to general SISO time-delay systems. Note that the FIR system in time-delay systems is generally not so easy to implement due to the sensitivity to the numerical errors. For the digital implementation (\ref{Coptform}), see \cite{ZhongTAC2004, ZhongTAC2005, MirkinSCL2004}. The designed optimal $\Hi$ controller for the admissible plants is a time-delay system. If a rational $\Hi$ controller is desired in a practical implementation, the optimal $\Hi$ controller can be approximated by a rational controller with a $\Hi$ performance level close to the optimal $\Hi$ norm, see \cite{OIJC91,TOIJRNC96,PARC04}.


\section{Numerical Examples} \label{sec:examples}
We compute optimal $\Hi$ performances and design optimal $\Hi$ controllers of the control problem~(\ref{eq:wsm1}) for SISO time-delay systems $P_1(s)$, $P_2(s)$, $P_3(s)$ in Example~\ref{ex:fact}. The weight functions for each plant are given in Table~\ref{table:problemdata}.
\begin{table}[h]
\begin{center}
\begin{tabular}{ccccccc}
  Plants & & $W_1(s)$ & & $W_2(s)$ & & $\gamma_{opt}$\\
  \hline \\[-6pt]
  $P_1(s)$ &  & $(0.1s+1)/(s+2)$ & & $0$          & & $1.8595$ \vspace{1mm} \\
  $P_2(s)$ &  & $(0.1s+1)/(s+2)$ & & $0.2(s+1.1)$ & & $0.9579$ \vspace{1mm}\\
  $P_3(s)$ &  & $(s+1)/(10s+1)$  & & $0.5$        & & \hspace{-1.2mm}$0.5534$\\
  \hline
\end{tabular}
\end{center}
\caption{The plants, weight functions and the corresponding optimal $\Hi$ performances for the control problem (\ref{eq:wsm1}).} \label{table:problemdata}
\end{table}

The coprime-inner/outer factorizations of plants $P_1(s)$, $P_2(s)$, $P_3(s)$  are obtained in Appendix~\ref{sec:app_fact}. Using the factorization terms and weights of each plant, the optimal $\Hi$ performances $\gamma_{opt}$ are computed by \cite{FOT,TOTAC95} and are given in the last column of Table~\ref{table:problemdata}. Having calculated $E_{\gamma_{opt}}(s)$, $F_{{\gamma}_{opt}}(s)$, $L_{{\gamma}_{opt}}(s)$, the optimal $\Hi$ controller can be written as (\ref{CoptFF}) for $P_1(s)$, $P_2(s)$ and as (\ref{CoptIF}) for $P_3(s)$. We eliminate unstable pole-zero cancelations in optimal $\Hi$ controllers by the decomposition operator $\Phi$ and obtain the transfer function in (\ref{Coptform}). The transfer functions of the terms $H_n(s)$, $F_n(s)$, $H_d(s)$, $F_d(s)$ in optimal $\Hi$ controllers for $P_1(s)$, $P_2(s)$, $P_3(s)$ are given in Appendix~\ref{sec:app_Copts}. The impulse responses of $F_n(s)$ and $F_d(s)$ for $P_1(s)$, $P_2(s)$, $P_3(s)$ have a finite support as  shown in Figure~\ref{fig:F}--\ref{fig:Fd3}. This illustrates that optimal $\Hi$ controllers have FIR structure for general SISO time-delay systems.

\begin{figure}[h]
\centering \includegraphics[width=8cm]{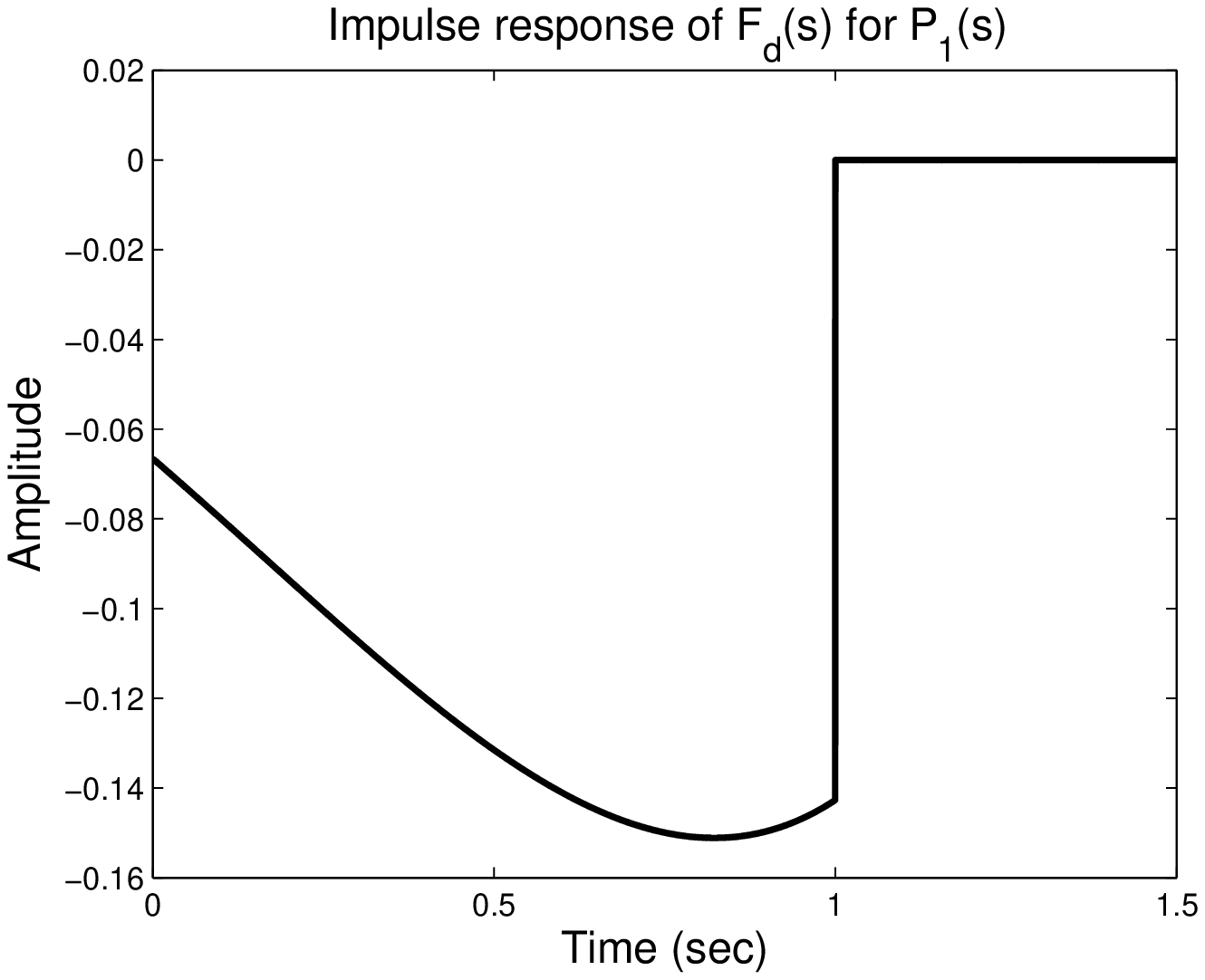}
\caption{The impulse response of $F_d(s)$ for $P_1(s)$} \label{fig:Fd1}
\begin{minipage}[b]{0.5\linewidth}
   \centering \includegraphics[width=6cm]{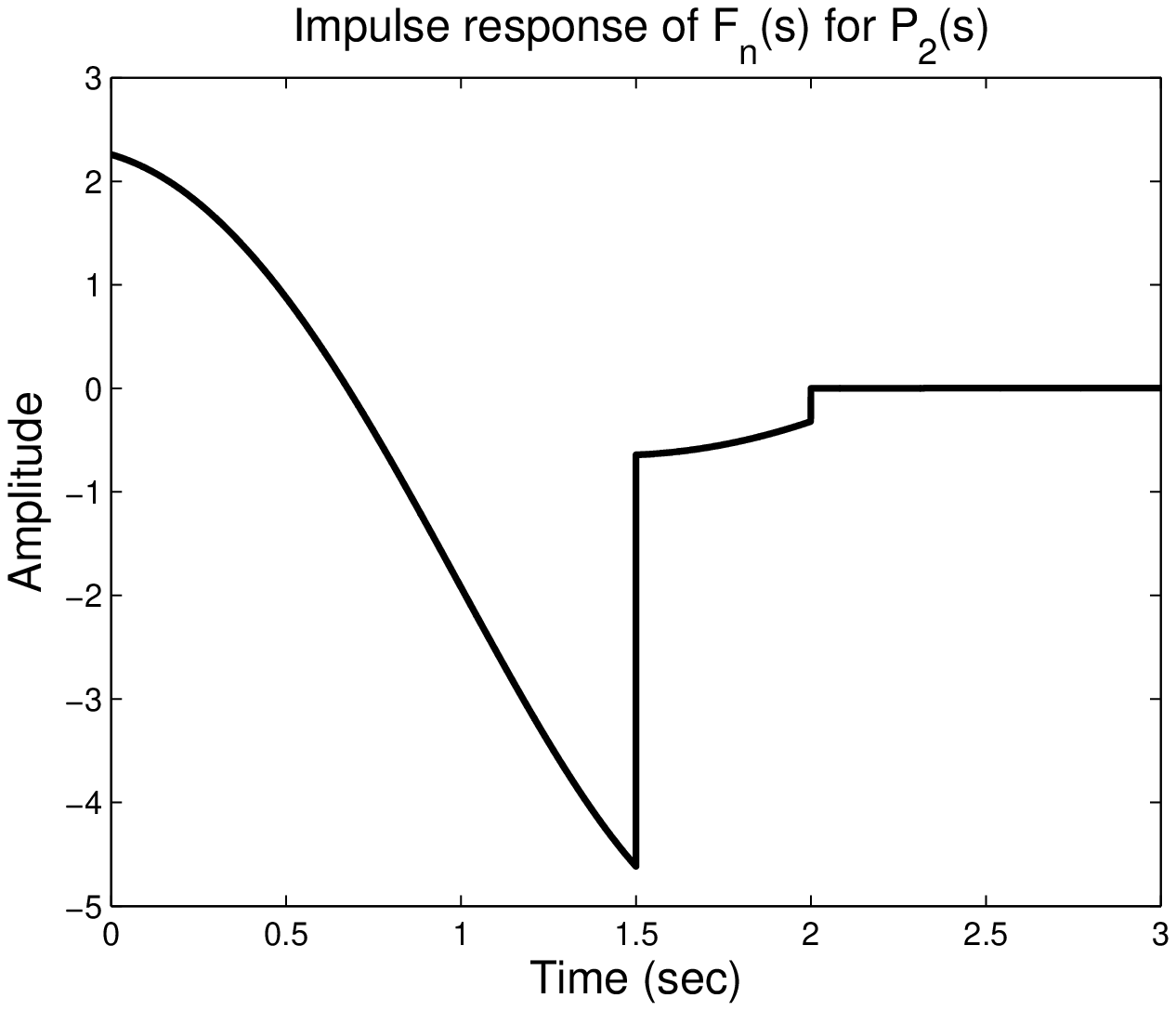}
   \caption{The impulse response of $F_n(s)$ for $P_2(s)$} \label{fig:Fn2}
\end{minipage}%
\begin{minipage}[b]{0.5\linewidth}
   \centering \includegraphics[width=6cm]{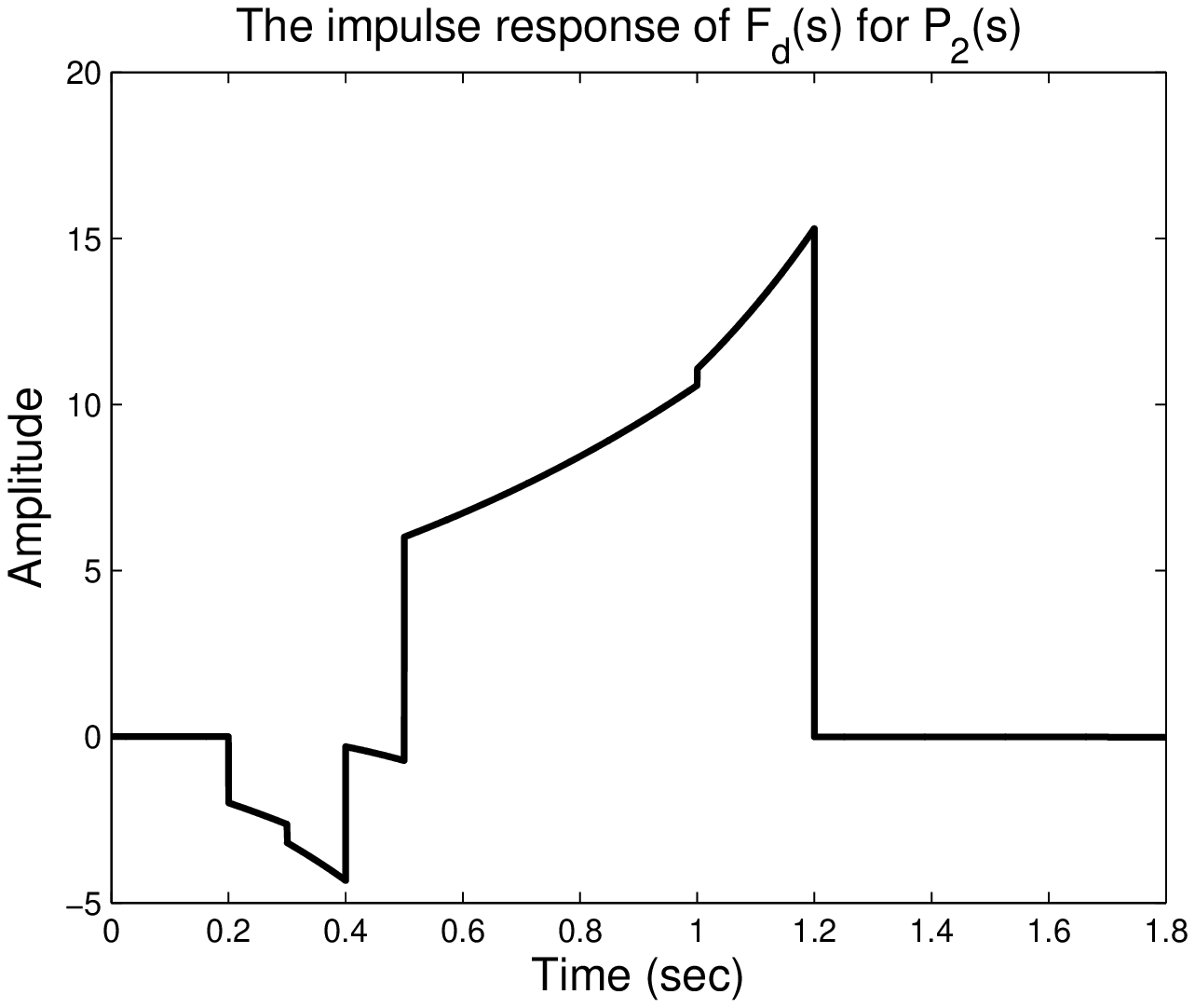}
   \caption{The impulse response of $F_d(s)$ for $P_2(s)$} \label{fig:Fd2}
\end{minipage}
\end{figure}
 \begin{figure}[h]
\begin{minipage}[b]{0.5\linewidth}
   \centering \includegraphics[width=6cm]{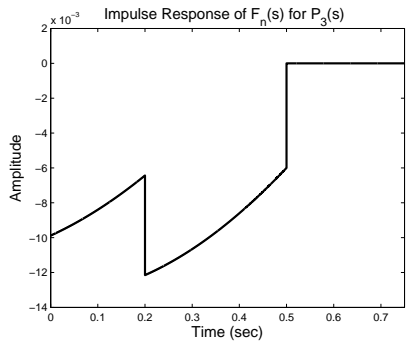}
   \caption{The impulse response of $F_n(s)$ for $P_3(s)$} \label{fig:Fn3}
\end{minipage}%
\begin{minipage}[b]{0.5\linewidth}
   \centering \includegraphics[width=6cm]{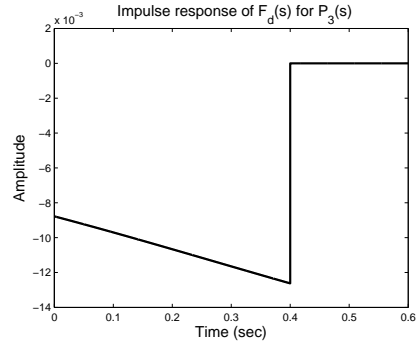}
   \caption{The impulse response of $F_d(s)$ for $P_3(s)$} \label{fig:Fd3}
\end{minipage}
\end{figure}

\section{Concluding remarks} \label{sec:conc}

We computed the optimal $\Hi$ performance and designed the optimal $\Hi$ controller for general SISO time-delay systems using the approach in \cite{FOT,TOTAC95}. We classified SISO plants admitting a coprime-inner/outer factorization. We gave the conditions to check the admissible plants and gave the explicit factorization terms. Using the factorization terms, we obtained optimal $\Hi$ controllers. We eliminated unstable pole-zero cancellations in these controllers and showed that the FIR structure of $\Hi$ controllers appears not only for I/O time-delay plants but general SISO time-delay plants. We illustrated our design on numerical examples. The analysis of the sensitivity of the coprime factorization for time-delay systems with respect to its time-delay and its applications to the computation of $\Hi$ performance are our future research directions.

\bibliographystyle{wileyj}
\bibliography{references}

\section{Appendix}
\subsection{Coprime-inner/outer factorizations of plants in Example~\ref{ex:fact}} \label{sec:app_fact}
\subsubsection*{Factorized plant $P_1$ by $C_1$ (\ref{Ccases})}:  The roots of $q_n(s)$ and $q_d(s)$ of $P_1(s)$ inside $\Cp$ are $1.0209 \pm1.4536j$ and $0.6235\pm0.8514j$ respectively. We obtain the inner functions $m_{q_n}(s)$ and $m_{q_d}(s)$ and factorize the plant $P_1(s)$ as in $C_1$ of (\ref{Ccases}): \\

\noindent $m_n(s)=\frac{s^2-2.0418s+3.1553}{s^2+2.0418s+3.1553}, \quad m_d(s)=\frac{s^2-1.2470s+1.1137}{s^2+1.2470s+1.1137}$,\\

\hspace{-4mm} $N_o(s)=\left(\frac{s^2-2s+3+0.2se^{-s}}{s^2-2.0418s+3.1553}\right)\left(\frac{s^2-1.2470s+1.1137}{s^3+1+e^{-1.5s}}\right)\left(\frac{s^2+2.0418s+3.1553}{s^2+1.2470s+1.1137}\right).$
\subsubsection*{Factorized plant $P_2$ by $C_1$ (\ref{Ccases})}: The quasi-polynomials $q_n(s)$ and $q_d(s)$ of $P_2(s)$ have $1$ and $2$ roots inside $\Cp$ at  $1.1296$ and $0.4153\pm1.6032j$ respectively. We compute the inner functions $m_{q_n}(s)$ and $m_{q_d}(s)$ from their roots inside $\Cp$ and factorize the plant $P_2(s)$ as in $C_1$ of (\ref{Ccases}): \\

\noindent $m_n(s)=\left(\frac{s-1.1296}{s+1.1296}\right)e^{-0.2s}, \quad m_d(s)=\frac{s^2-0.8306s+2.7426}{s^2+0.8306s+2.7426},$ \\

\hspace{-4mm} $N_o(s)=\left(\frac{(s-1)+(0.1s+1)e^{-0.1s}+(0.2s-3)e^{-0.8s}}{s-1.1296}\right)
\left(\frac{s^2-0.8306s+2.7426}{3s+0.5+(2s+7)e^{-1.5s}+(s-1)e^{-2s}}\right)$ \\
${}\hspace{11.3cm}\left(\frac{s+1.1296}{s^2+0.8306s+2.7426}\right).$
\subsubsection*{Factorized plant $P_3$ by $C_2$ (\ref{Ccases})}: The  roots of $\bar{q}_n(s)$ and $q_d(s)$ of the plant $P_3(s)$ inside $\Cp$ are $0.2470$ and  $0.4672\pm1.8891j$ respectively. We get the inner functions $m_{\bar{q}_n}(s)$ and $m_{q_d}(s)$ and factorize the plant $P_3(s)$ as in $C_2$ of (\ref{Ccases}): \\

\hspace{-4.5mm} $m_n(s)=\left(\frac{s-0.2470}{s+0.2470}\right)\left(\frac{s+3+(2s-2)e^{-0.4s}}{2s+2+(s-3)e^{-0.4s}}\right), \quad m_d(s)=\frac{s^2-0.9343s+3.7868}{s^2+0.9343s+3.7868},$  \\

\hspace{-4mm} $N_o(s)=\left(\frac{2s+2+(s-3)e^{-0.4s}}{s-0.2470}\right) \left(\frac{s^2-0.9343s+3.7868}{s^2+se^{-0.2s}+5e^{-0.5s}}\right) \left(\frac{s+0.2470}{s^2+0.9343s+3.7868}\right).$

\subsection{Optimal $\Hi$ controllers computed in numerical examples} \label{sec:app_Copts}

The optimal $\Hi$ controller terms for the plants $P_1(s)$, $P_2(s)$ and $P_3(s)$ are given below. Unstable pole-zero cancelations are transformed into the FIR terms, $F_n(s)$ and $F_d(s)$.

\subsubsection*{Optimal $\Hi$ controller terms for the plant $P_1$}: \\

\noindent $H_n=\frac{{\scriptsize \begin{array}{r}
(204.8139s^4+990.1165s^3+1729.2310s^2+1328.9154s+384.9869)+
(0.5589s^4+9.1606s^3 \\
+26.3922s^2+30.6421s+20.0340)e^{-1.5s}
\end{array}}}
{s^5+15s^4+59s^3+97s^2+72s+20},$ \\

\noindent $F_n=\frac{(-0.1260s+0.3061)-(0.5588s+0.0810)e^{-1.5s}}{s^2-1.2470s+1.1137},$ \\

\noindent $H_d=\frac{{\scriptsize \begin{array}{r}
(17.7128s^5+23.9061s^4+7.7671s^3-123.9613s^2-150.7115s-120.3186)
-(0.1427s^5 \\
-1.4834s^4-4.8317s^3-10.2271s^2-7.7671s-3.9171)e^{-s}
\end{array}}}
{s^6+15.0418s^5+62.6992s^4+139.4s^3+177.423s^2+118.234s+31.5533},$

\medskip

\noindent $F_d=\frac{(-0.1260s+0.3061)-(0.5588s+0.0810)e^{-1.5s}}{s^2-1.2470s+1.1137}.$

\subsubsection*{Optimal $\Hi$ controller terms for the plant $P_2$}: \\

\noindent $H_n=\frac{{\scriptsize \begin{array}{r}
2.5576s^3-23.7178s^2-66.2867s-27.6938-(0.7643s^3+26.4543s^2+104.3494s \\
+103.2943)e^{-1.5s}+(1.2847s^3-6.0435s^2-11.5679s+5.6143)e^{-2s}
\end{array}}}
{s^4+6.7782s^3+3.7362s^2-21.2389s-19.1968},$ \\

\noindent $F_n=\frac{(2.2581s-2.8559)+(3.9747s+0.6515)e^{-1.5s}+(0.3206s-1.3992)e^{-2s}}
{s^2-0.8306s+2.7426},$

\medskip

\noindent $H_d=\frac{
  {\scriptsize \begin{array}{l}
    (0.0333s^5+0.2933s^4+0.6379s^3-0.1744s^2-2.3605s-1.9407)e^{-0.2s}
    -(0.0150s^5 \\
    +0.1320s^4+0.2871s^3-0.0785s^2-1.0622s -0.8733)e^{-0.3s}
    -(5.6529s^5+48.4261s^4 \\
    +185.6168s^3+379.5343s^2+412.3743s+180.6268)e^{-0.4s}
    -(6.8948s^5+79.4156s^4+323.2374s^3 \\
    +634.8263s^2+574.2295s+187.8953)e^{-0.5s}
    +(0.0533s^5+0.4693s^4+1.0207s^3-0.2791s^2 \\
    -3.7768s-3.1051)e^{-s}+(14.9809s^5+180.1371s^4+738.4147s^3+1443.4060s^2+1274.2323s\\
    \hspace{10.2cm}+396.1757)e^{-1.2s}
  \end{array}}
}{s^6+9.7994s^5+27.9381s^4+13.9061s^3-76.0517s^2-129.0422s-58.2230},$

\medskip

\noindent $F_d=\frac{
  {\scriptsize \begin{array}{l}
    (-1.9920s^4-2.0032s^3-0.6709s^2-3.4737s+8.3747)e^{-0.2s}
    -(0.5516s^4+2.0532s^3+4.3838s^2 \\
    +6.3426s+5.1627)e^{-0.3s}+(4.0299s^4
    +1.1908s^3+0.1515s^2+8.6321s-18.2225)e^{-0.4s} \\
    +(6.7325s^4-1.5256s^3+4.3559s^2+22.2652s-39.4726)e^{-0.5s}
    +(0.4985s^4+4.3158s^3 \\ +10.8537s^2+14.5658s+16.9481)e^{-s}-(15.3054s^4-4.4247s^3+11.0188s^2+52.7515s \\
    \hspace{10cm} -92.1927)e^{-1.2s}
  \end{array}}
}{s^5-1.9602s^4+0.7387s^3+2.6692s^2-10.8299s+9.1153}.$

\subsubsection*{Optimal $\Hi$ controller terms for the plant $P_3$}: \\

\noindent $H_n=\frac{
  {\scriptsize \begin{array}{l}
    10^{-2}(1.4687s^3+0.7746s^2+0.1283s+6.9708\ 10^{-3})
    -10^{-3}(5.3127s^3-2.8339s^2+9.2829\ 10^{-2}s \\
    \hspace{0.8cm} -1.57435\ 10^{-2})e^{-0.2s}
    -10^{-3}(5.9993s^3
    +39.9491s^2-5.3170\ 10^{-2}s+2.0087\ 10^{-2})e^{-0.5s}
  \end{array}}
}{s^4+2.2312s^3+1.4755s^2+0.2575s+0.01312},$ \\

\noindent $F_n=\frac{(-0.009883s+0.02164)-(0.005714s+0.004544)e^{-0.2s}
+(0.005999s-0.03418)e^{-0.5s}}
{s^2-0.9343s+3.7868},$ \\

\noindent $H_d=\frac{
  {\scriptsize \begin{array}{l}
   (-0.5004s^7+6.2571s^6+5.9356s^5-0.1183s^4-14.4458s^3-4.0624s^2-0.7188s-0.099362) \\
   -(0.1610s^7+16.5703s^6+30.9780s^5
   +11.6068s^4-4.4790s^3-1.8178s^2-0.3662s \\
   \hspace{10cm} -0.006552)e^{-0.4s}
  \end{array}}
}{10^{3}(s^8+4.4514s^7+7.9042s^6+7.1200s^5+3.4667s^4+0.9528s^3+0.1740s^2+0.02196s+0.001250)},$ \\

\noindent $F_d=\frac{(-0.8777s^4+0.1428s^3-3.1172s^2-0.1909s-0.1546)+
(1.2623s^4-0.5283s^3+3.6099s^2+0.2144s+0.1640)e^{-0.4s}}
{10^2(s^5-1.1813s^4+4.0410s^3-0.9630s^2+0.0941s-0.02191)}.$ \\
\end{document}